\documentclass[runningheads,a4paper]{llncs}

\usepackage{amssymb}
\usepackage{graphicx}

\usepackage{url}
\urldef{\mailsa}\path|{alfred.hofmann, ursula.barth, ingrid.haas, frank.holzwarth,|
\urldef{\mailsb}\path|anna.kramer, leonie.kunz, christine.reiss, nicole.sator,|
\urldef{\mailsc}\path|erika.siebert-cole, peter.strasser, lncs}@springer.com|

\usepackage{ltl}
\usepackage{wrapfig}
\usepackage{booktabs}
\usepackage{tikz}
\usepackage{color}
\usepackage{scalefnt}
\usepackage{makecell}
\usepackage{array}

\newcommand{\todo}[1]{}

\newcommand{\AP}{\textit{AP} }
\renewcommand{\next}{\LTLcircle}

\newcommand{\nats}{\mathbb{N}}

\newcommand{\LTL}{\textsc{LTL}}	

\newcommand{\srate}{\mathcal r}
\newcommand{\usrate}{\mathcal r^\infty}
\newcommand{\bp}{\textit{Bad}}
\newcommand{\gp}{\textit{Good}}

\newcommand{\card}[2]{\#_{#1}(#2)}

\newcommand{\nlogspace}{\textsc{NL}}

\newcommand{\pspace}{\textsc{Pspace}}

\newcommand{\exptime}{\textsc{Exptime}}

\begin{document}

\mainmatter  % start of an individual contribution

% first the title is needed
\title{The Density of Linear-time Properties%
  \footnote{This work was partly supported by the ERC Grant 683300 (OSARES)
     and by the Deutsche Telekom Foundation.}}

% a short form should be given in case it is too long for the running head
\titlerunning{The Density of Linear-Time Properties}

% the name(s) of the author(s) follow(s) next
%
% NB: Chinese authors should write their first names(s) in front of
% their surnames. This ensures that the names appear correctly in
% the running heads and the author index.
%
\author{Bernd Finkbeiner \and Hazem Torfah}
\authorrunning{}
% (feature abused for this document to repeat the title also on left hand pages)

% the affiliations are given next; don't give your e-mail address
% unless you accept that it will be published
\institute{Saarland University}

%
% NB: a more complex sample for affiliations and the mapping to the
% corresponding authors can be found in the file "llncs.dem"
% (search for the string "\mainmatter" where a contribution starts).
% "llncs.dem" accompanies the document class "llncs.cls".
%

\toctitle{Lecture Notes in Computer Science}
\tocauthor{Authors' Instructions}
\maketitle

\begin{abstract}
Finding models for linear-time properties is a central problem in verification and planning. We study the distribution of linear-time models by investigating the density of linear-time properties over the space of ultimately periodic words.
 The density of a property over a bound~$n$ is the ratio of the number of lasso-shaped words
of length $n$, that satisfy the property, to the total number of lasso-shaped words of length~$n$. We investigate the problem of computing the density for both linear-time properties in general and for the important special case of $\omega$-regular properties. For general linear-time properties, the density is not necessarily convergent and can oscillates indefinitely for certain properties. However, we show that the oscillation is bounded by the growth of the sets of bad- and good-prefix of the property. For $\omega$-regular properties, we show that the density is always convergent and provide a general algorithm for computing the density of $\omega$-regular properties as well as more specialized algorithms for certain sub-classes and their combinations.
\end{abstract}

\section{Introduction}
Given a linear-time property, specified for example as a formula of a temporal logic, how hard is it to guess a model of the property? Temporal models play a fundamental role in verification
and planning, for example in the satisfiability problem of temporal logic~\cite{Rozier+Vardi/07/LTL}, in model checking~\cite{Baier:2008:PMC:1373322}, and in
temporal planning~\cite{DBLP:conf/ijcai/PatriziLGG11}. With this paper, we initiate the first systematic study of the \emph{density} of the linear-time temporal models. 

The first choice to be made at the outset of such an investigation is how
to represent temporal models. We base our study on ultimately
periodic words, i.e., infinite words of the form $u \cdot v^\omega$,
where $u$ and $v$ are finite words. This is motivated by the fact that
ultimately periodic words are the natural and commonly used
representation in all applications, where the underlying state space is
finite (cf. \cite{Clarke:2001:BMC:510986.510987}).  With this
choice of representation, our central
question thus is the following:
Suppose you are given an infinite word $u \cdot v^\omega$, where $u$
and $v$ are two finite words, 
that have been chosen randomly from all
sequences over a given alphabet. How likely is it that $u\cdot v^\omega$ is a model of a given
linear-time property?

We consider the \emph{cardinality} and the \emph{density function} in terms of the bound~$n$.
The cardinality of a property $\varphi$ for a given bound $n$ is the number of lassos of length $n$, that are models for the property $\varphi$, denoted by 
$\card{\varphi}{n} = | \{ (u,v) \in \Sigma^{*}\times \Sigma^{+} \mid |u \cdot v| = n,~ u\cdot v^\omega \in \varphi\} |$.  The density function of a property $\varphi$ determines the rate of the cardinality of $\varphi$  over the whole solution space for the specific bound $n$, and is denoted by $\srate_\varphi(n) = \card{\varphi}{n} / | \{ (u,v) \in \Sigma^{*}\times \Sigma^{+} \mid |u \cdot v| = n \} | = \card{\varphi}{n} / (n\cdot |\Sigma|^n)$. To answer the question posed above, we study the asymptotic behavior of the density function, i.e., the limit $\lim \limits_{n \rightarrow \infty} \srate_\varphi(n)$, which we denote by $\usrate_\varphi$ and refer to as the \emph{density} of the property $\varphi$.

Consider
 the following linear-time properties over the alphabet $\Sigma= 2^{\{a,b\}}$. The density function of the property given by the LTL formula $\varphi_1 = a \wedge \LTLcircle b \wedge \LTLcircle \LTLcircle b$ is constant and equal to $\frac{1}{8}$ for bounds larger than two, because there is no restrictions on the labeling once the constraint of a labeling $a$ followed by two $b$'s has been satisfied. 
The density of the property $\varphi_2 = a \LTLuntil b$ is equal to $\frac{2}{3}$ as the increase in the number of models is twice as large as the increase in the number of non-models for increasing bounds. Properties like  $\varphi_3 = \LTLsquare (a \wedge \LTLcircle b)$ have densities equal to 0, because the cardinality of the set of bad-prefixes increases exponentially with increasing bounds, in comparison to a linear increase in the number of its models. 

Two key questions of interest are whether or not the density exists for a linear-time property and, if the answer to the first question is yes, to compute its value.
It is not obvious that the density exists for linear-time properties. 
In the case of $\omega$-regular properties, we show that the density indeed always exists. This stands in contrast to regular properties of finite words, where this is not always true. 
Consider, for example, the regular property $(aa)^*$. Models for the property exist only for even bounds, and the density function oscillates between 0 and 1 for the alphabet $\{a\}$.  The $\omega$-regular property $(aa)^\omega$ for the same alphabet, however, has ultimately periodic models for all bounds and its density function converges to~0.
The density for linear-time properties in general, nevertheless, does not necessarily exist. We show that for certain not $\omega$-regular properties, the density function oscillates indefinitely without converging.

In case the density function cannot be computed, we show that it can still be approximated by examining the growth of the sets of good- and bad-prefixes of the property. The density of good-prefixes of a property defines a lower bound on the density. The density of bad-prefixes defines an upper bound on the density. Whether a density exists for property $\varphi$ depends on the densities of four classes of lassos, that partition the whole space of lassos with respect to $\varphi$. These classes represent lassos $(u,v)$, where $u\cdot v$ is a bad-prefix for $\varphi$, a good-prefix for $\varphi$, or models or non-models of $\varphi$ where $u\cdot v$ is neither a good- nor a bad-prefix. We present few ways to check  the existence of the density of a linear-time property with respect to the densities of each of these classes. To illustrate the affect of these classes consider the property $\next\next\next a$. It is clear that the rate of both the classes of models with no good-prefix and non-models with no bad-prefix converge to~0. This means, the upper and the lower bound of the density, determined by the classes of bad-prefix non-models and good-prefix models, meet in the limit and the density function converges to $\frac{1}{2}$.  

For the special case of $\omega$-regular properties, the limit of the density function can be computed algorithmically.
This can be done by constructing an unambiguous $\omega$-automaton, that defines the property and computing the probability of reaching an accepting strongly connected component in the automaton.  
Building on top of the algorithmic ideas, we also investigate the qualitative density checking problems, i.e., we determine if the density of a property is equal to 0 or 1, and provide a complete complexity analysis, determining the lower and upper bounds of these problems. Table~\ref{tab:summary} gives a summary on the complexity results shown and proven in the paper. 
%A gap between the lower (\pspace) and upper bound (\exptime) remains open for the problem of determining whether the density is smaller than 1 for a given LTL formula.

\begin{table}[h]
\caption{Results for the computational complexity of computing the density of $\omega$-regular languages.}
\centering
\begin{tabular}{|c||c|c|c|}
	\hline 
	 & LTL & non-deterministic B\"uchi & deterministic Parity \\
	\hline
	\hline
	$\usrate_\varphi$ & \exptime & \exptime & P \\
	\hline
	 $\usrate_\varphi > 0$ & \pspace-compl. & P-compl. & \nlogspace-compl. \\
	\hline
	$\usrate_\varphi < 1$ & \pspace-compl. &  \pspace-compl. & \nlogspace-compl. \\
	\hline
\end{tabular}
\label{tab:summary}
\end{table}

For some sub-classes of $\omega$-regular properties, we can even avoid the costly construction of the automaton. We investigate a series of sub-classes and show how to compute the density for these classes and any of their combinations. In the case of \LTL, we match syntactic classes to the introduced sub-classes and show that the density of a boolean combination of these syntactic classes can be reduced to the computation of the density of a much smaller formula. 

\paragraph{\bf Related work.}
In the setting of finite, rather than infinite, words, the study of
density has a long history~\cite{Bodirsky+others/04/Efficiently,Chomsky195891,Flajolet1987283,Hartwig/10/Density,Szilard+others/92/Characterizing}. For each language $\varphi \subseteq \Sigma^{*}$ of finite words over some alphabet $\Sigma$, the \emph{density function} is defined as the quotient $\srate_\varphi(n) =
\card{\varphi}{n}/|\Sigma^n|$, where $|S|$ denotes the cardinality of a set
$S$ and $\card{\varphi}{n} = |\varphi \cap
\Sigma^n|$, i.e., the number of words of length $n$ in $\varphi$.
 In 1958, Chomsky and Miller~\cite{Chomsky195891} showed that for each regular language $\varphi$,
there exists an initial length $n_0$ such that for all $n\geq n_0$,
$\card{\varphi}{n}$ can be described by a linear recurrence. For example, for the language $\psi$ of the regular expression $(ab+baa)^{*}$, we have that $\card{\psi}{n} = \card{\psi}{n-2}+\card{\psi}{n-3}$. The recursive description of $\card{\varphi}{n}$ allows for a detailed analysis of the shapes of $\card{\varphi}{n}$ and $\srate_\varphi(n)$ (cf.~\cite{Hartwig/10/Density}).
The result
was later extended to the nonambiguous context-free languages~\cite{Hartwig/10/Density}. Much attention has focussed on \emph{sparse} languages, i.e., languages, where $\#_\varphi(n)$ can be bounded from above by a polynomial~\cite{Demaine2003471,Flajolet1987283,Szilard+others/92/Characterizing}. Sparse languages can be used to restrict NP-complete problems so that they can be solved polynomially~\cite{Demaine2003471}. An interesting application of the density is to determine how well a non-regular language is approximated by a finite automaton~\cite{Eisman:2005:ARN:1082161.1082186}; this is important in streaming algorithms, where the incoming string must be classified quickly, and it suffices if the classification is correct most of the time.

    In previous work~\cite{DBLP:conf/lata/FinkbeinerT14}, we have presented automata-based algorithms for computing $\card{\varphi}{n}$ for safety specifications expressed in LTL. These algorithms compute $\card{\varphi}{n}$ for a specific property $\varphi$ and a specific value of $n$, but cannot be used to derive the convergences value of $\srate_\varphi(n)$ for an entire class of properties.
  Faran and Kupferman have recently investigated the probability that a prefix of a word not in $\varphi$ is a bad prefix of $\varphi$ \cite{FK15}. This probability is used to quantitatively determine the ``safety level'' of $\varphi$. The analysis again is done with respect to the finite words not the infinite words. A key difference to our work is that the safety level does not give the probability of a model and does not distinguish between properties of the same class but with different density values.
 Also related is Asarin \emph{et~al.}'s investigation of the asymptotic behavior in temporal logic~\cite{Asarin:2014:ABT:2603088.2603158}. The authors use the notion of entropy to show the relation between formulas in parametric linear-time temporal logic and formulas in standard LTL as some bounds tend to infinity.

%%%%%%%%%%%%%%%%%%%%%%%%%%%%%%
%%%%%%%%%%% Prelim %%%%%%%%%%%

\section{Preliminaries}
\paragraph{Linear-time Properties and Models.} A \emph{linear-time property} $\varphi$  over an alphabet~$\Sigma$ is a set of infinite words $\varphi \subseteq \Sigma ^\omega$. 
Elements of $\varphi$ are called \emph{models} of $\varphi$. The complement set $\overline\varphi=  \Sigma^\omega \setminus \varphi$ is called the set of \emph{non-models} of $\varphi$. 

A \emph{lasso} over an alphabet $\Sigma$ of length $n$ is a pair $(u,v)$ of finite words $u\in \Sigma^{*}$ and $v \in \Sigma^{+}$ with  $|u\cdot v|~=n$, that induces the ultimately periodic word  $u \cdot v^\omega$. We call $u\cdot v$ the base of the lasso or ultimately periodic word. 
An $n$-\emph{model} for the property $\varphi$  over $\Sigma$ is a lasso $(u,v)\in \Sigma^* \times \Sigma^+$ of length~$n$ such that the induced ultimately periodic word $u\cdot v^\omega \in \varphi$. We call $n$ the bound of the $n$-model. 
 The language $ L_n(\varphi)$ for a bound $n$ is set of $n$-models of~$\varphi$. Note that a model of $\varphi$ might be induced by more than one $n$-model, e.g, $a^\omega$ is induced by $(a, a)$ and $(\epsilon, a a)$. The complement language  $\overline{ L_n(\varphi)}$ is the set of $n$-non-models of $\varphi$. 
We call the linear-time property over $\Sigma$, whose models build up the set of all lassos over $\Sigma$ the \emph{universal property} and denote it by $\top$. 
The \emph{cardinality} of a property $\varphi$ for a bound $n$, denoted by $\card{\varphi}{n}$, is the size of  $ L_n(\varphi)$.  

\paragraph{Safety and Liveness.} For an infinite word $\sigma = \alpha_1\alpha_2 \dots \in \Sigma^\omega$ we denote every prefix $\alpha_1\dots \alpha_i$ by $\sigma[\dots i]$.
A finite word $w= \alpha_1 \dots \alpha_i \in \Sigma^{*}$ is called a bad-prefix for a property $\varphi$, if every infinite word $\sigma \in \Sigma^\omega$ with $\sigma[\dots i]= w$ is not a model of $\varphi$. We call a bad-prefix $w$ \emph{minimal}, if no prefix of $w$ is a bad-prefix for $\varphi$. We denote the set of bad-prefixes of a property $\varphi$ by $\bp(\varphi)$. 
A finite word $w =\alpha_1 \dots \alpha_i \in \Sigma^{*}$ is called a good-prefix for a property $\varphi$, if every infinite word $\sigma \in \Sigma^\omega$ with $\sigma[\dots i]=w$ is a model of $\varphi$. We call a good-prefix $w$ minimal, if $w$ has no prefix, that is also a good-prefix for $\varphi$. We denote the set of good-prefixes of a property $\varphi$ by $\gp(\varphi)$.

A property $\varphi$ is a \emph{safety} property if every non-model of $\varphi$ has a bad-prefix for $\varphi$.  A property $\varphi$ is a \emph{liveness} property if  every finite word $w$ can be extended by an infinite word $\sigma$ such that $w \cdot \sigma$ is a model of $\varphi$. 
A property $\varphi$ is a \emph{co-safety} property if every model of $\varphi$ has a good-prefix for $\varphi$. Co-safety properties can be either safety or liveness properties. The only property that is both liveness and safety at the same time is the \emph{universal} property $\top$.

\paragraph{Linear-time Temporal Logic.}
We use Linear-time Temporal Logic (LTL) \cite{Pnueli:1977:TLP:1382431.1382534}, with the usual temporal operators Next $\LTLcircle$, Until $\LTLuntil $, and the derived operators Release $\LTLrelease$, which is the dual operator of $\LTLuntil$, Eventually $\LTLdiamond$ and Globally~$\LTLsquare$.
 LTL formulas are defined over a set of atomic propositions $\AP$.
We denote the satisfaction of an LTL formula $\varphi$ by an infinite sequence $\sigma \in (2^{AP})^\omega $ of valuations of the atomic propositions  by $\sigma \models \varphi$ and call $\sigma$ a model of $\varphi$. For an LTL formula~$\varphi$ we define $L(\varphi)$ by the set $\{\sigma \in (2^{AP})^\omega \mid \sigma \models \varphi \}$. A lasso $(u,v)$ of length $n$ is an $n$-model of an LTL formula $\varphi$ if $u\cdot v^\omega \in L(\varphi)$. If $u\cdot v^\omega$ is not a model of $\varphi$, the lasso is  called an $n$-non-model of $\varphi$. 

\paragraph{Parity Automata.}
A parity automaton over an alphabet $\Sigma$ is a tuple $\mathcal A = (Q,Q_0,\delta,c)$, where $Q$ is a set of states, $Q_0$ is a set of initial states, $\delta: Q \times \Sigma \rightarrow 2^Q$ is a transition relation, and $c: Q \rightarrow \nats$ is a coloring function. 
A run of $\mathcal A$ on an infinite word $w = \alpha_1 \alpha_2 \dots \in \Sigma^\omega$ is an infinite sequence $r = q_0 q_1 \dots \in Q^\omega$ of states, where $q_0 \in Q_0$ and for each $i \ge 0$, $q_{i+1} = \delta(q_i,\alpha_{i+1})$. 
We define $\textbf{Inf}(r)=\{q \in Q \mid \forall i \exists j>i.~q_j = q\}$. A run $r$ is called accepting if $\max \{c(q) \mid q \in \textbf{Inf}(r)\}$ is even. A word $w$ is accepted by $\mathcal A$ if there is an accepting run of $\mathcal A$ on $w$. 

 The automaton is called deterministic if the set $Q_0$ is a singleton and for each $(q,\alpha) \in Q \times \Sigma$ we have  $|\delta(q,\alpha)|\leq 1$. The automaton is called unambiguous if for each accepted word $w$ there is exactly one accepting run of the automaton on~$w$. A parity automaton is called a \emph{B\"uchi} automaton if the image of $c$ is contained in $\{1,2\}$. An automaton is \emph{complete} if each state has an outgoing transition for each letter $\alpha \in \Sigma$. In the paper we always consider complete automata. 

A strongly connected component (SCC) in $\mathcal A$  is a strongly connected component of the graph induced by the automaton.
 A strongly connected component is called 
 \emph{terminal} if none of the states in the SCC has a transition, that leaves the SCC. An SCC is called \emph{accepting} if the highest color of the states of the SCC is even.

%%%%%%%%%% END PRELIM %%%%%%%%%%%
%%%%%%%%%%%%%%%%%%%%%%%%%%%%%%%%%

%%%%%%%%%%%%%%%%%%%%%%%%%%%%%%%%%
%%%%%%%%%%% Density %%%%%%%%%%%%%
\section{The Density of Linear-time Properties}

For a given linear-time property $\varphi$, the density function of $\varphi$ gives the  distribution of models of $\varphi$ for increasing bounds $n$.

\begin{definition}[Density]
\label{def:srate}
	The density function of a linear-time property $\varphi$ over an alphabet $\Sigma$ and a bound $n$ is the ratio between the cardinality of $\varphi$ for $n$ and the number of lassos of length $n$ over $\Sigma$: 
$$
	\srate_{\varphi}(n) =\frac{\card{\varphi}{n}}{n\cdot |\Sigma|^n}
$$

The asymptotic density of $\varphi$ (short density) is the value $\lim \limits_{n \rightarrow \infty}~ \srate_\varphi(n)$ (in case it exists), which we  denote by $\usrate_{\varphi}$.
\label{def:Density}
\end{definition}

In previous work we presented an algorithm for computing the cardinality of safety LTL formulas for a given bound. The algorithm is doubly-exponential in the length of the formula yet linear in the bound \cite{DBLP:conf/lata/FinkbeinerT14}. An algorithm based on a translation to a propositional formula is exponentially less expensive in the formula than our algorithm, but exponentially more expensive in the bound. With respect to counting complexity classes, the complexity of computing the density of a property $\varphi$ depends on the complexity of the membership test allowed by the representation of $\varphi$. Counting the number of models for a bound $n$ and a property given as an LTL formula has been shown to be in $\#$P \cite{TZ14}. Using these results we summarize the counting complexities of computing the density of $\omega$-regular properties in the following theorem\footnote{ For more on counting complexities and the counting problem for linear-time temporal logic we refer the reader to \cite{DBLP:conf/lata/FinkbeinerT14,TZ14,DBLP:journals/tcs/Valiant79}}. 

\begin{theorem}
\label{theo:sharpnlogspaceauto}
	For an $\omega$-regular property $\varphi$ given by an LTL formula, a nondeterministic B\"uchi automaton, or a deterministic parity automaton, and for a given bound $n$, the problem of computing $\srate_\varphi(n)$ is in $\#$P.
\end{theorem}
\begin{proof}
	To show that the problem is in \#P, we show that there is nondeterministic polynomial-time Turing machine $\mathcal M$, such that, the number of accepting runs of the machine on a given bound $n$ and a property $\varphi$, is equal to the number of $n$-models of $\varphi$.  We define $\mathcal M$ as follows. The machine $\mathcal M$ guesses a prefix $u$ and a period $v$ of an ultimately periodic word $u\cdot v^\omega$ with $|u\cdot v| = n$, and checks whether $u\cdot v^\omega$ satisfies $\varphi$, which can be done polynomial time when $\varphi$ is an LTL formula \cite{KuhtzFinkbeiner09}, and in logarithmic space when $\varphi$ is given by a nondeterministic B\"uchi or a deterministic parity automaton \cite{Markey2003}. For each $n$-model $(u, v)$ of $\varphi$ there is exactly one accepting run of $\mathcal M$. Thus, counting the $n$-models of $\varphi$ can be done by counting the accepting runs of $\mathcal M$ on the input $(n, \varphi)$. \qed
\end{proof}

Before getting to the computational complexity of computing the density of a given linear-time property, we illustrate what factors play a role in shaping the density function of a property. Consider the density of the universal property $\top$, which is constant and equal to 1, as its cardinality is defined by $\card{\varphi}{n}=n\cdot|\Sigma|^n$. For each bound $n$, we can transform every $n$-model of $\top$ to a $(n+1)$-model by extending the base of the $n$-models with one letter from $\Sigma$ and adding one of the now $n+1$ possible loops to the new base. The number of bases for $n$-models for the property $\top$ is $\frac{\card{\top}{n}}{n}$. Thus, the number of $(n+1)$-models for $\top$ is  equal to $|\Sigma|\cdot\frac{n+1}{n}\cdot \card{\top}{n}$. According to the definition of the density function, this means that the monotonicity of the density function of a property in some bound $n$ depends on whether the increase in the number of models in $n$ is larger or smaller than the growth factor $|\Sigma|\cdot\frac{n+1}{n}$. 

We define the \emph{growth function} of a property $\varphi$ by $ \varsigma_\varphi(n) = \frac{\#_{\varphi}(n+1)}{\#_\varphi(n)}$. We call the function $\varsigma_{\top}(n)=|\Sigma|\cdot\frac{n+1}{n}$ the \emph{universal growth function}. The following proposition clarifies the relation between the monotonicity of the density function of a linear-time property $\varphi$ and the universal growth function. Furthermore, the proposition shows the relation between the growth function of $\varphi$ and the growth function of its complement $\overline \varphi$.

\begin{proposition}
\label{prop:rateCompare}
Given a property $\varphi$ over $\Sigma$ the following holds:
\begin{enumerate}
	\item $\forall n.~\srate_\varphi(n) = \srate_\varphi(n+1) $ if and only if $ \varsigma_\varphi(n) = \varsigma_\top(n)$
	\item $\forall n.~\srate_{\varphi}(n+1) > \srate_{\varphi}(n) $ if and only if $  \varsigma_{\varphi}(n) > \varsigma_{\top}(n)$
	\item $\forall n.~\varsigma_{\varphi}(n) = \varsigma_\top(n) $ if and only if $ \varsigma_{\overline\varphi}(n) = \varsigma_\top(n)$
	\item $\forall n.~\varsigma_{\varphi}(n) > \varsigma_\top(n) $ if and only if $ \varsigma_{\overline \varphi}(n) <\varsigma_\top(n)$
\end{enumerate}
\end{proposition}

For any linear-time property $\varphi$, the function $\card{\varphi}{n}$ is monotonically increasing. This is due to the fact that any $n$-model of the property $\varphi$ can be mapped to a $(n+1)$-model of $\varphi$, namely the one that results from unrolling the loop of the $n$-model by one position. From the last proposition we read thus that the monotonicity of the density function of $\varphi$ at some bound $n$ depends on the number of new $(n+1)$-models, i.e., those that cannot be rolled back to $n$-models. 
A density function, where the increase in models at some bound is higher (lower) than the increase in all lassos (models of the universal property $\top$) is  increasing (decreasing) at that bound. This in turn means that the growth factor of the number of non-models is lower (higher) at the same bound. An oscillating function is one, where the increase in the number of models is interchangeably higher and lower than the increase in the total number of lassos. 

Whether the density of a property exists, converges, is monotone or oscillating, depends on the densities of the following classes of lassos, that form with respect to a given property a partition of the space of lassos. 
For a property $\varphi$,  we split the set of lassos into four classes:

\begin{itemize}
	\item \textbf{Base non-models:} These are non-models $(u,v)$, where $u\cdot v \in \bp(\varphi)$.   
	\item \textbf{Base models:} These are models $(u,v)$, where $u\cdot v \in \gp(\varphi)$. 
	\item \textbf{Loop non-models:} These are non-models $(u,v)$, where $u \cdot v \not \in \bp(\varphi)$. 
	\item \textbf{Loop models:} These are models $(u,v)$ where $u \cdot v \not \in \gp(\varphi)$. 
\end{itemize}

\begin{figure}[t]
\centering
\includegraphics[width=0.8\textwidth]{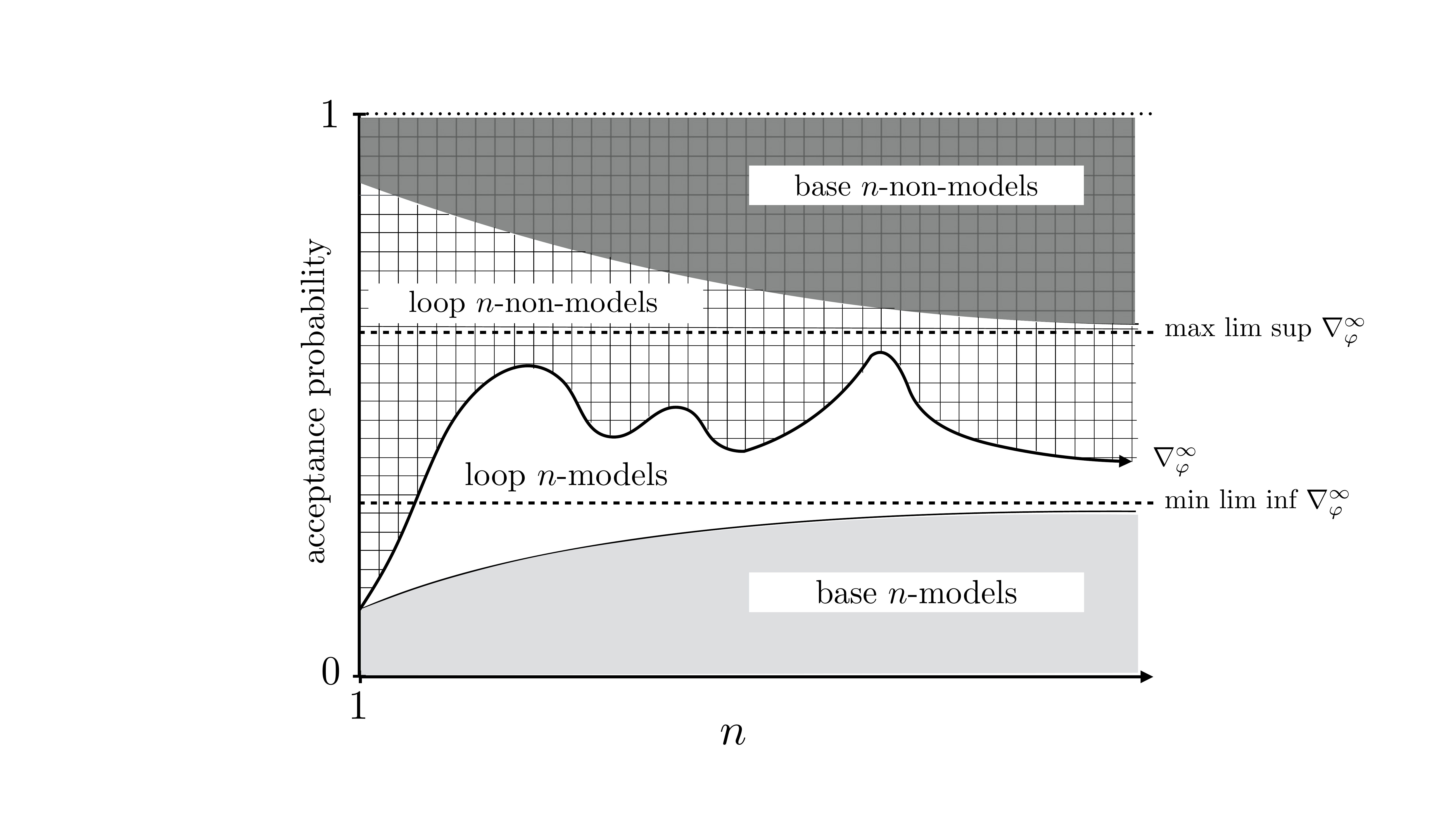}
\caption{The change in the density of the different classes of lassos for a linear property $\varphi$ over increasing bounds $n$. Notice that both the rates of base $n$-non-models (lined gray area) and base $n$-models (plain gray area) is monotonically increasing forming an upper and lower bound on the density.}
\label{fig:classes}
\end{figure}

In Figure~\ref{fig:classes}, we show how each of these classes grow over increasing bounds. For any property $\varphi$ and for all bounds $n$, the classes of base non-models and base models increase by a factor larger or equal to $\varsigma_\top(n)$, because any extension of a bad-prefix remains a bad-prefix and every extension of a good-prefix also remains a good-prefix. Following Proposition \ref{prop:rateCompare}, this means that for any bound~$n$, the rates of base $n$-models and base $n$-non-models to the set of all $n$-lassos are monotonically increasing and thus converging. The rate of base models defines for each $n$ a lower bound for the density function $\srate_\varphi(n)$. Its convergence value defines in turn a lower bound on the limit inferior of the density function. The rate of base non-models defines an upper bound on the density function and its convergence value is an upper bound on the limit superior of the density function. 
The increase in the number of lassos of the classes of loop models and loop non-models depends highly on the property. An extension of the bases may result in a new bad-prefix a new good-prefix, or a base on top of which new loop models or non-models can be obtained. This  means that the rate of these two classes to the set of all lassos might oscillate indefinitely without converging as we show for some  properties in the next section. This in turn means that the convergences behavior of the density function of $\varphi$ is determined by the convergence of the rate of loop models of $\varphi$.

%%%%%%%%%% END Density %%%%%%%%%%%
%%%%%%%%%%%%%%%%%%%%%%%%%%%%%%%%%%

%%%%%%%%%%%%%%%%%%%%%%%%%%%%%%%%%%%%%%%%%%%%
%%%%%%%%%%% Asymptotic Density %%%%%%%%%%%%%

\subsection{Asymptotic Density}
In this section we investigate which linear-time properties have a converging density function. In the case of finite regular properties the density does not always exist. This follows from the fact that some regular properties allow no models for certain bounds as we have seen in the introduction. In contrast, in the case of $\omega$-regular properties we will show that the density always exists. For general linear-time properties however we will see that this does not necessarily hold when considering  $\omega$-non-regular properties as shown in detail in Theorem~\ref{theo:nonconvprop}.

We classify a property $\varphi$ according to the convergence of its density function to either: \emph{0-convergent} when $\usrate_\varphi = 0$, \emph{1-convergent} when $\usrate_\varphi = 1$, \emph{$\epsilon$-convergent} when $\usrate_\varphi = \epsilon$ for ~$0<\epsilon<1$, and \emph{$\bot$-convergent} when the density function is non-convergent. 
 
The change in the size of the different classes of lassos presented in the last section plays a key role in the convergence behavior of a property. From the last section we know that the rates of base models and base non-models are always convergent. This means the convergence behavior depends on the rates of loop models and loop non-models. 
For example,
 the property $\next p$, has no loop models nor loop non-models for bounds larger that 2, and the rates of these classes converge to 0. All lassos of length greater or equal to 2 belong to one of the sets of base models or base non-models, depending on whether the second position of the lasso is labeled with $p$ or not, and thus, the density of $\next p$ is determined by the rates of base models and base non-models. The number of base models of $\next p$ is equal to $2^{\AP-1}\cdot(2^{\AP})^{n-1}\cdot n$ for $n>1$, which results in a density of $\frac{1}{2}$.
 
 The rates of base models and base non-models also  determine the density of the safety property $q \LTLrelease p$ over $\AP=\{p,q\}$, which  convergences to a value of $\frac{1}{3}$. 
   The property has no loop non-models and $n$ loop models for each bound $n$, namely those where all positions are labeled with $p$ and not labeled with $q$. Thus, the rates of loop models and loop non-models converge to 0. In the case of base-models we can count $\sum \limits_{i=1}^{n} (2^{\AP})^{n-i} \cdot n$ many base $n$-models, because for each $1 \leq i \leq n$, there are $(2^{\AP})^{n-i} \cdot n$ many base $n$-models which are labeled with $p$ and $q$ at position $i-1$ and with only $p$ for all positions smaller than $i$, and arbitrarily for all positions greater than $i$. Applying Definition~\ref{def:Density}, the density function of $q \LTLrelease p$ can be computed as $\sum \limits_{i=1}^n (4)^{-i}$, which converges towards $\frac{1}{3}$ when $n$ tends to infinity.
   
An example, where the density depends fully on the rate of loop models is $\LTLdiamond \LTLsquare p$. The property has neither base models nor base non-models. A lasso is a loop model for $\LTLdiamond \LTLsquare p$ if all positions in the loop are labeled with $p$. For a bound~$n$, there are $\sum \limits_{i=1}^{n} (2^\AP)^{i-1}$ many loop $n$-models ($i$ is the position of the loop). This results in a density function equal to $\sum \limits_{i=1}^{n} \frac{ (2^{\AP})^{i-1}}{(2^\AP)^{n} \cdot n}$ which converges to 0 when $n$ grows to infinity\footnote{The formula $\LTLdiamond \LTLsquare p$ is an example of a 0-convergent liveness formula.}.
   
If none of the sets of bad- nor good-prefixes is empty, and the rate of loop models is convergent then the density is $\epsilon$-convergent, because none of the rates of base models nor base non-models to all lassos is equal to zero.    

\begin{lemma}
	The density function of a property $\varphi$ is convergent, if and only if the rate of loop models is convergent.
\label{lem:loopconv}
\end{lemma}
\begin{proof}
	The density function can be defined as the sum of the two rates of base models and loop models. Because the rate of base models is always convergent, it follows that the density function is convergent if and only if the rate of loop models is convergent. \qed
\end{proof}
With the same argumentation the rate of loop non-models plays the same role as the one for loop models.
 
\begin{lemma}
For a given property $\varphi$, the rate of loop models is convergent if and only if the rate of loop non-models is convergent. 	
\end{lemma}
\begin{proof}
	From Lemma~\ref{lem:loopconv} we know that when the rate of loop models is convergent then the density of $\varphi$ is also convergent. This means that the rate of non-models is convergent, and as the rate of base non-models is always convergent, then so is the rate of loop non-models. 
	
	With  analogous  reasoning we can show that the convergence of the rate of loop non-models implies the convergence of the rate of loop models. \qed
\end{proof}

We show now an example of certain types of  not $\omega$-regular properties, that have a non-convergent rate of loop models, and thus a non-convergent density function. 

\begin{theorem}
\label{theo:nonconvprop}
	There is a linear-time property with a non-convergent density function. 
\end{theorem}
\begin{proof}
	We divide each of the lasso classes further into cyclic lassos, i.e., lassos, where the loop is at the first position of the lasso, and non-cyclic lassos, which cover the rest of lassos in a class. The property we present is a liveness property, where the classes of base models and base non-models are empty. We show that the rate of loop models is non-convergent and  we show that the reason is that the rate of cyclic loop models is non-convergent.
	
	We define a property $\varphi$  over the set of atomic propositions $\AP =\{a\}$ as follows:
        Let $c_1, c_2 \dots$ and $d_1, d_2 \dots$ be natural numbers such that $c_1 \leq d_1 <c_2 \leq d_2< \dots$ . A lasso is a model of $\varphi$ if eventually the letter $\{a\}$ is encountered at some position and there is a constant $\delta$ in one of the intervals $[c_i,d_i)$ for some  number $i \in \nats$ such that from then on, $\{a\}$ appears periodically every $\delta$ positions. 
        
        The number of non-cyclic loop $n$-models of $\varphi$ is equal to $|2^\AP|\cdot \#_\varphi(n-1)$, because we can extend each $(n-1)$-model $\varphi$  to an $n$-model by attaching any letter from $2^{\AP}$ to the first position of the $(n-1)$-model. 
       This means that the growth factor of the models of $\varphi$ is determined by the respective growth in the size of the sets of cyclic models for each bound. 
       
       Notice that the rate of cyclic models depends strongly on the  bound $n$. If $c_j \leq n < d_j$ for some $j$, then we have $|2^\AP|^n -1$ many cyclic models of length $n$, as we only need to have at least one position labeled with $\{a\}$ and the rate of models in this case is increasing. If $d_j \leq n <c_{j+1}$, then there are at most $|2^\AP|^{n-\lfloor{\frac{n}{h}}\rfloor}\cdot h$, where  $h$ is the largest allowed period in $n$. In this case, the rate of models is decreasing. This means that in all bounds $n$ that belong to some interval $[c_i,d_i)$ the density is increasing, and all bounds in intervals $[d_j,c_{j+1})$ the density of $\varphi$ is decreasing. 
       We can choose the numbers $c_1,c_2, \dots$ and $d_1,d_2 \dots $ in a way that the density increases in $c_{i+1}$ to a value larger than the one in $c_i$, and decreases in $d_{i+1}$ to a value smaller than in $d_i$. In this way, the density function is oscillating and non-convergent. 
\qed
\end{proof}

%%%%%%%%%%% END Asymptotic Density %%%%%%%%%%%%%
%%%%%%%%%%%%%%%%%%%%%%%%%%%%%%%%%%%%%%%%%%%%%%%%

%%%%%%%%%%%%%%%%%%%%%%%%%%%%%%%%%%%%%%%%%%%%%%%%%%%
%%%%%%%%%%% Density Regular Properties %%%%%%%%%%%%
\subsection{Density of $\omega$-Regular Properties}
In this section we show how  to compute the density of $\omega$-regular properties given by non-deterministic B\"uchi and  deterministic parity automata. In the next section we show how these results can be adopted to compute the density of properties given as  LTL formulas.  

We start by showing the relationship between the density of a property $\varphi$ and the densities of the terminal SCCs of an  automaton representing $\varphi$. 

\begin{lemma}
The density of an $\omega$-regular property $\varphi$ given as a parity automaton $\mathcal A$ is greater than 0 if and only if $\mathcal A$ has a reachable terminal accepting strongly connected component.
\label{lem:maxSCC}  	
\end{lemma}
\begin{proof}
We prove the lemma along the steps of \cite{Finkbeiner+Schewe/06/ProbEnv}. Let $\mathcal A$ be defined over the alphabet $2^\AP$ for a set of atomic proposition $\AP$. Let $S$ be an accepting terminal SCC in $\mathcal A$ with $n$ states (remember that $\mathcal A$ is complete, thus $S$ allows transitions for all letters in $2^\AP$ in each state).
The probability of choosing a transition with label $\alpha \in 2^\AP$ is equal to $\epsilon = \frac{1}{|2^\AP|}$ from any state in $S$. Let $s$ be a state in $S$. The probability of not reaching $s$ from any other state in $S$ in $n$ steps is at most  $1-\epsilon^n<1$. This means for an infinite trace in $S$, the probability of not seeing  $s$ again from every position of the trace is equal to 0. Thus, the probability of choosing an infinite run  $\sigma$ in $S$ such that $s \not \in \textbf{Inf}(\sigma)$ is also equal to 0. This  is in particular true for the state $s_{\max}$ with the maximum color in $S$.  This implies that an  infinite run $\sigma$ in $S$ with $\textbf{Inf}(\sigma)$ equal to the set of states of $S$, has probability 1. 
Because $s_{\max}$ is even,  it follows that the probability of a lasso with an accepting run in $S$ converges to 1, when the length of the lasso tends to infinity.    

If $\mathcal A$ has a reachable terminal accepting SCC $S$, then it is reached by at least one finite prefix with positive probability. The density of $\varphi$ is then at least equal to the probability of choosing this prefix.
If $\mathcal A$ does not have a terminal accepting SCC, then the rate of models of $\varphi$  converges to 0, because  the probability of staying infinitely in an accepting cycle in the automaton is 0. \qed 
\end{proof}

Using the previous lemma we show the complexity of the following qualitative problems for the density.
\begin{theorem}
\label{theo:denlarg0}
Let $\varphi$ be an $\omega$-regular property given by a deterministic parity automaton. 
	The problem of checking whether $\usrate_\varphi>0$ is \nlogspace-complete.  
\end{theorem}
\begin{proof}
	As shown in Lemma~\ref{lem:maxSCC}, to check if the density of the property given by an automaton $\mathcal A$ is greater than 0, it suffices to check whether there is an accepting terminal strongly connected component in $\mathcal A$. 	We choose a state $q$ of $\mathcal A$ reachable from the initial state and apply the following procedure.  Iterating over all states $q'$ of the automaton, we check if $q'$ is reachable from $q$. If yes, we check if $q$ is reachable again from $q'$. If this is not the case, then $q$ is not a state in a terminal SCC in $\mathcal A$, and we choose a new state in $\mathcal A$ different than $q$  and repeat the whole procedure for the new state. Otherwise, if for each $q'$ reachable from $q$, there is a path leading back to $q$, then we have found a terminal SCC that contains $q$.  During the iteration we also save the highest color seen. If this color is even then $q$ is a state in a terminal accepting SCC in $\mathcal A$.	
	  If no terminal accepting SCC is found after iterating over all states of $\mathcal A$, then the density  is 0. 
	  
	Checking whether a state is reachable from another can be done in nondeterministic logarithmic space (the reachability problem in automata is in \nlogspace). Checking whether a state is not reachable from another can also be done in non-deterministic logarithmic space (the non-reachability problem in automata is in co-\nlogspace ~and  \nlogspace=co-\nlogspace).  In each iteration we only need to memorize the binary encoding of the  state $q$ and the current state $q'$ and the current highest color seen so far, which require in total an encoding of no more than logarithmically many bits in the size of the automaton.
	
	A matching lower bound is can be proven by a reduction from a nondeterministic logarithmic-space Turing machines. \qed
\end{proof}

\begin{theorem}
	Let $\varphi$ be an $\omega$-regular property given by a nondeterministic  B\"uchi automaton $\mathcal A$. 
	The problem of checking whether $\usrate_\varphi>0$ is P-complete.  
\end{theorem}
\begin{proof}
	We check whether $\mathcal A $ has a terminal accepting SCC $S$. Finding such an SCC can be done in polynomial time.
A matching lower bound is achieved by a log-space reduction from the \emph{circuit value} problem.\qed
\end{proof}

We turn now to the problem of checking if $\usrate_\varphi <1$.

\begin{theorem}
\label{theo:densm1deter}
Let $\varphi$ be an $\omega$-regular property given by a deterministic parity automaton $\mathcal A$. 
	The problem of checking whether $\usrate_\varphi<1$ is \nlogspace-complete.

\end{theorem}
\begin{proof}
Following the idea of Theorem~\ref{theo:denlarg0} we can check if $\mathcal A$ has a terminal \textbf{non-accepting} SCC in nondeterministic logarithmic space. Because the automaton is deterministic, any lasso that has a run in this SCC is a non-model of $\varphi$. If $\mathcal A$ contains such an SCC, then the rate of non-models is greater than 0 (at least as equal as the probability of reaching the non-accepting SCC), and thus the density is less than 1. If no such SCC is found, then the density is equal to~1. 

A matching lower bound can be shown via a reduction from a nondeterministic logarithmic-space Turing machine. 
\end{proof}

\begin{theorem}
\label{theo:densm1nondeter}
Let $\varphi$ be an $\omega$-regular property given by a nondeterministic B\"uchi automaton $\mathcal A$. 
	The problem of checking whether $\usrate_\varphi<1$ is \pspace-complete.
\end{theorem}
\begin{proof} 
Using  Safra's construction  \cite{1691236}, every nondeterministic parity automaton $\mathcal A$ can be transformed into a deterministic parity automaton $\mathcal D$ of size exponential in the size of $\mathcal A$. Each state of $\mathcal D$ is a Safra-tree over the states of $\mathcal A$. The size of a Safra-tree is equal to the size of $\mathcal A$ and we can distinguish exponentially many Safra-trees. 
In Theorem~\ref{theo:densm1deter} we presented a non-deterministic logarithmic-space algorithm over deterministic parity automata for checking whether there is an non-accepting terminal SCC. Instead of constructing the whole automaton $\mathcal D$ and checking the existence of such an SCC, we will do it on the fly as follows. 
 We can guess a run of the automaton $\mathcal D$ by stepwise guessing Safra-trees and checking if two succeeding trees are consistent with transition relation of $\mathcal D$. At some position we also guess that a current state $q$ of the run is one in a terminal SCC. As in the procedure of Theorem~\ref{theo:denlarg0} we check whether all successor states $q'$ allow a path from which $q$ can be reached again. For that we only need logarithmic space in the size of $\mathcal D$, thus, polynomial space in the size of $\mathcal A$. If a maximum color seen during the traversal of a path is even, then we have found an accepting terminal SCC that contains the state $q$.
	
	A matching lower bound can be achieved following the steps of \cite{Meyer:1972:EPR:1437899.1438639} by reducing a polynomial space-bound Turing machine $\mathcal M$ and a word $w$ to a non-deterministic parity automaton $A$ such that, $\mathcal M$ accept $w$ if and only if the density of $A$ is smaller 1. 
\end{proof}

Using the so far presented results we show now how to compute the density.
\begin{theorem}
	Computing the density $\usrate_{\varphi}$ for an $\omega$-regular property $\varphi$ can be done in polynomial time if $\varphi$ is given by an unambiguous parity automaton $\mathcal A$.
\label{theo:compconv}
\end{theorem}
\begin{proof}
To compute the density of $\varphi$ we need to compute the density of each terminal accepting strongly connected component in the automaton $\mathcal A$. Because $\mathcal A$ is unambiguous, it is guaranteed that no model ends in two terminal accepting SCCs of $\mathcal A$ and thus the density is the sum of densities of all terminal accepting SCCs. The density of an SCC is given  by the probability of reaching the SCC. 
	Computing the probability of an SCC can be seen as a convergence problem of a Markov chain, where the automaton $\mathcal A$ can be thought of as a Markov chain, where the label of a  transition is replaced by its probability, i.e., a probability equal to $\frac{1}{|2^{AP}|}$. Both finding the terminal accepting  strongly connected components and computing their probabilities can be done in polynomial time \cite{Bodirsky+others/04/Efficiently}. \qed 
\end{proof}

\section{Density of LTL Properties}

In this section we reexamine the problems investigated in the last section for properties given as LTL formulas. For any LTL property $\varphi$ we can compute the density by constructing an unambiguous $\omega$-automaton for $\varphi$ and using the algorithm given in Theorem~\ref{theo:compconv}. However, the construction of the automaton is costly (exponential \cite{Baier:2008:PMC:1373322}) and can be avoided for many sub-classes of LTL. 

The results for LTL are summarized below, and follow from Theorem~\ref{theo:compconv}, \ref{theo:densm1nondeter}, and \ref{theo:denlarg0}, and from the fact that any LTL formula can be turned into an exponential unambiguous parity automaton \cite{Baier:2008:PMC:1373322}. 
\begin{theorem}
\begin{enumerate}
	\item Computing the density $\usrate_{\varphi}$ for an $\omega$-regular property $\varphi$ given as an LTL formula can be done in exponential time.
	\item Checking whether $\usrate_\varphi>0$ for an LTL formula $\varphi$ is \pspace-complete.
	\item Checking whether $\usrate_{\varphi} < 1$ for an LTL formula $\varphi$ is \pspace-complete.
\end{enumerate}
\end{theorem}

In the following we present a series of syntactic LTL classes for which the density can be immediately given. Using these sub-classes we show later that the computation of the density for LTL formulas can be reduced to the computation of the density of a much smaller LTL formula. We distinguish following syntactic LTL classes:

\begin{itemize}
	\item \textbf{Bounded-Safety:} A \emph{bounded-safety} property $\varphi$ describes a set of infinite words, each with a prefix in a finite set $\Gamma \subseteq \Sigma^{k}$ for some $k$. A formula in the LTL fragment with only the temporal operator $\LTLnext$ is a bounded-safety formula.
	\item \textbf{Invariants:} An \emph{invariant} property $\varphi$ describes an unreachability property over a bounded-safety property $\psi$ and is given by the LTL fragment $\LTLsquare \psi$.
	\item \textbf{Guarantee:} A \emph{guarantee} property $\varphi$ is a reachability property defined over some bounded-safety property $\psi$ and is given by the LTL fragment $\LTLdiamond \psi$.
	\item \textbf{Persistence:} A \emph{persistence} property $\varphi$ is a co-B\"uchi condition defined over a bounded-safety property $\psi$ and is given by the LTL fragment $\LTLdiamond \LTLsquare \psi$.
	\item \textbf{Response:} A \emph{response} property $\varphi$ is a B\"uchi condition defined over a bounded-safety property and is given by the LTL fragment $\LTLsquare \LTLdiamond \psi$.
\end{itemize}

\begin{theorem}
\label{theo:subclassconv}
\begin{enumerate}
	\item Every bounded-safety property $\varphi$ not equivalent to false or to true is $\epsilon$-convergent. 
	\item Every invariant property or persistence property $\varphi$ is 0-convergent.
	\item Every guarantee property or response property $\varphi$  is 1-convergent.
\end{enumerate}
\end{theorem}

\subsection{Composition of LTL Properties}
When given an LTL formula $\varphi$ composed of formulas from the syntactic classes presented above we can compute the density of $\varphi$ using the rules given in Table~\ref{tab:usrateconj}. The intersection of properties $\varphi_1$ and $\varphi_2$ that are convergent to 1 results in a new property that also has an density $\usrate_{\varphi_1 \cap \varphi_2} = 1$. When $\varphi_1$ and $\varphi_2$ converge to 0 then $\usrate_{\varphi_1 \cap \varphi_2} = 0$. The same also holds when considering the union of the properties $\varphi_1$ and $\varphi_2$. In the case of $\epsilon$-convergent properties $\varphi_1$ and $\varphi_2$ the density of the intersection of the properties depends intersected properties. If both $\varphi_1$ and $\varphi_2$ were bounded-safety properties this value depends on the size of the intersection of characterization sets of $\varphi_1$ and $\varphi_2$. It can range from $0$, when the characterization sets are disjoint, to $\epsilon$ when the properties are equivalent. When building the union of two $\epsilon$-convergent properties the density can range from $\epsilon$ to 1. If both properties were again bounded safety properties then the density is equal to $\epsilon$ when $\varphi_1$ and $\varphi_2$ are equivalent and to 1 when their characterization sets are disjoint.
\begin{table}[t]
\caption{density for conjunctive (lower triangle) and disjunctive (upper triangle) compositions:}
\centering
\begin{tabular}{|c||c|c|c||c|}
	\hline 
	$\usrate_{\varphi_1 \cap \varphi_2}$ & 1 & $\epsilon$& 0 &$\usrate_{\varphi_1 \cup \varphi_2}$\\
	\hline
	\hline
	1 & 1 & 1 &1  &1\\
	\hline
	$\epsilon$ & $\epsilon$ & \diaghead(-1,1){\hskip 0.9cm}{0/$\epsilon$}{1/$\epsilon$}  & $\epsilon$ &$\epsilon$\\
	\hline
	0 & 0 & 0& 0 &0\\
	\hline
\end{tabular}
\label{tab:usrateconj}
\end{table}
Given an LTL formula $\varphi$ composed of the syntactic classes we  apply the rules presented in Table~\ref{tab:usrateconj} and the results from Theorem~\ref{theo:subclassconv},  until no rule is applicable anymore. The remaining formula is a bounded safety formula for which we apply the algorithm given in Theorem~\ref{theo:compconv}.

For example, consider the LTL formula over the set of atomic propositions $\AP= \{a,b\}$:
$$\varphi = (a \vee \next b) \wedge (\next \next \next (b \wedge a) \vee \LTLdiamond a) \vee (\LTLsquare b \wedge \LTLdiamond (a \wedge \next b))$$

\noindent We start by evaluating the subformulas  
$$\usrate_{\LTLsquare b \wedge \LTLdiamond (a \wedge \next b)} = \usrate_{\LTLsquare b}=0 \text{~~~~and~~~~}  \usrate_{\next \next \next (b \wedge a) \vee \LTLdiamond a} = 1$$
  Thus the density is equal to the one of the formula $(a \vee \next b)$, which is a bounded-safety property for which we can use the algorithm in Theorem~\ref{theo:compconv} and compute the density 
 $$\usrate_{a \vee \next b}= 0.5 + 0.5*0.5= 0.75$$

%%%%%%%%%%%%%%%%%%%%%%%%%%%%%%%%%%%%%%%%%%%%%
%%%%%%%%%%%%%%% Discussion %%%%%%%%%%%%%%%%%%

\section{Discussion}

With this paper, we have initiated an investigation of the density of models of linear-time properties. Our work extends the classic results for finite words to ultimately periodic infinite words.
In comparison to finite words, the new class of models significantly complicates the analysis; the proof techniques introduced in this paper, in particular the analysis of classes of loop and base  models and non-models, have allowed us, however, to obtain a classification of the major property classes according to the convergence of the density.
Computing the density for omega-regular properties can be done algorithmically, yet is very expensive. 
In contrast to expensive LTL algorithms presented above, the qualitative analysis can be obtained for free, for the syntactic fragments for 
the different property types introduced in the paper (and their combinations). 

The obvious next step is to exploit the results algorithmically. 
It may be possible to steer randomized algorithms such as Monte Carlo model checking~\cite{Grosu/05/Monte} towards areas of the solution space where we are most likely to find a model. 
In planning, the choice between exploration and backtracking in a 
temporal planner could be biased towards exploration in 
situations with increasing probability, and towards backtracking 
in situations with decreasing probability. It may also be possible to develop approximative algorithms that replace a complicated linear-time property with a simpler, but ultimately equivalent property, such as a parity condition with a smaller number of colors. In similar techniques for properties of finite words, the density of the difference language is used to verify that the error introduced by the approximation is small~\cite{Eisman:2005:ARN:1082161.1082186}.

A big challenge is to extend the results further to tree models and, thus, to determine the density of branching-time properties. A first step into this direction is made by model counting algorithms for tree models~\cite{DBLP:conf/lata/FinkbeinerT14}. Since tree models can be seen as implementations in the sense of reactive synthesis~\cite{DBLP:conf/fossacs/Thomas09}, this line of work might also lead to a better understanding of the complexity of the synthesis problem, and perhaps to new randomized synthesis algorithms.

\bibliographystyle{plain}
\bibliography{biblio.bib}

\end{document}